\documentclass[a4paper,USenglish]{lipics-v2019}
\usepackage{hyperref}


\usepackage{graphicx,comment}
\usepackage[linesnumbered,vlined]{algorithm2e}

\usepackage{amsmath}
\usepackage{amsthm}
\usepackage{amsfonts, amssymb}
\usepackage{dsfont}
\usepackage{cite,booktabs}
\usepackage{xcolor}

\newcommand{\etal}{\textit{et al.}\xspace}
\newenvironment{claim-repeat}[1]{\begin{trivlist}
		\item[\hspace{\labelsep}{\bf\noindent Claim \ref{#1} }]\em }%
	{\end{trivlist}}
\newenvironment{theorem-repeat}[1]{\begin{trivlist}
		\item[\hspace{\labelsep}{\bf\noindent Theorem \ref{#1} }]\em }%
	{\end{trivlist}}

\newcommand{\eps}{\epsilon}
\newcommand{\Expc}[1]{\mathbb{E}\left[{#1}\right]}
\newcommand{\size}[1]{\ensuremath{\left|#1\right|}}
\newcommand{\set}[1]{\left\{ #1 \right\}}
\DeclareMathOperator{\roundp}{round}
\newcommand{\Gcal}{\mathcal{G}}
\newcommand{\Hcal}{\mathcal{H}}
\newcommand{\bbN}{\mathbb{N}}
\newcommand{\bbR}{\mathbb{R}}
\newcommand{\ceil}[1]{\lceil #1 \rceil}
\newcommand{\floor}[1]{\lfloor #1 \rfloor}
\newcommand{\gold}{G_{\mathrm{old}}}
\newcommand{\gtmp}{G_{\mathrm{adv}}}
\newcommand{\gnew}{G_{\mathrm{new}}}
\newcommand{\etmp}{E_{\mathrm{adv}}}

\newcommand{\gallowed}{\Gcal_{\mathrm{allowed}}}
\newcommand{\ind}[1]{\mathds{1}_{#1}}

\DeclareMathOperator{\polylog}{polylog}

\newcommand{\hide}[1]{}

\title{Models of Smoothing in Dynamic Networks}

\author
{Uri Meir}
{Tel Aviv University, Israel}
{}
{}
{}

\author
{Ami Paz}
{Faculty of Computer Science, Universit\"at Wien, Austria}
{}
{}
{Supported by the Austrian Science Fund (FWF): P 33775-N, Fast Algorithms for a Reactive Network Layer.}

\author
{Gregory Schwartzman}
{Japan Advanced Institute of Science and Technology, Japan}
{}
{}
{This work was supported by JSPS Kakenhi Grant Number JP19K20216 and~JP18H05291.} 

\authorrunning{U. Meir, A. Paz, and G. Schwartzman}

\Copyright{Uri Meir, Ami Paz, and Gregory Schwartzman}

\acknowledgements{The authors are thankful to Seth Gilbert for interesting discussions.}
\nolinenumbers
\hideLIPIcs  

\begin{document}
\maketitle

\begin{abstract}
Smoothed analysis is a framework suggested for mediating gaps between worst-case and average-case complexities. In a recent work, Dinitz et al.~[Distributed Computing, 2018] suggested to use smoothed analysis in order to study dynamic networks. 
Their aim was to explain the gaps between real-world networks that function well despite being dynamic, and the strong theoretical lower bounds for arbitrary networks.
To this end, they introduced a basic model of smoothing in dynamic networks, where an adversary picks a sequence of graphs, representing the topology of the network over time, and then each of these graphs is slightly perturbed in a random manner.

The model suggested above is based on a per-round noise, and our aim in this work is to extend it to models of noise more suited for multiple rounds. This is motivated by long-lived networks, where the amount and location of noise may vary over time.
To this end, we present several different models of noise.
First, we extend the previous model to cases where  the amount of noise is very small.
Then, we move to more refined models, where the amount of noise can change between different rounds, e.g., as a function of the number of changes the network undergoes. We also study a model where the noise is not arbitrarily spread among the network, but focuses in each round in the areas where changes have occurred. Finally, we study the power of an adaptive adversary, who can choose its actions in accordance with the changes that have occurred so far.
We use the flooding problem as a running case-study, presenting very different behaviors under the different models of noise, and analyze the flooding time in different models.
\end{abstract}

\section{Introduction}
Distributed graph algorithms give a formal theoretical framework for studying networks. There is abundant literature on distributed graph algorithms modeling \emph{static} networks, and recently, an increasing amount of work on \emph{dynamic} networks.
Yet, our understanding of the interconnections between the theoretical models and real world networks is unsatisfactory. That is, we have strong theoretical lower bounds for many distributed settings, both static and dynamic, yet real-world communication networks are functioning, and in a satisfactory manner.

One approach to bridging the gap between the strong theoretical lower bounds and the behaviour of real-world instances, is the idea of \emph{smoothed analysis}.
Smoothed analysis was first introduced by Spielman and Teng~\cite{SpielmanT09, SpielmanT04}, in an attempt to  explain the fact that some problems admit strong theoretical lower bounds, but in practice are solved on a daily basis.
The explanation smoothed analysis suggests for this gap, is that lower bounds are proved using very specific, pathological instances of a problem, which are highly unlikely to happen in practice. They support this idea by showing that some lower bound instances are extremely \emph{fragile}, i.e., a small random perturbation turns a hard instance into an easy one.
Spielman and Teng applied this idea to the simplex algorithm, and showed that, while requiring an exponential time in the \emph{worst case}, if we apply a small random noise to our instance before executing the simplex algorithm on it, the running time becomes polynomial in expectation.





Smoothed analysis was first introduced to the distributed setting in the seminal work of Dinitz \etal~\cite{DFGN18}, who considered distributed dynamic networks. These are networks that changes over time, and capture many real world systems, from Blockchain, through vehicle networks, and to peer-to-peer networks. They are modeled by a sequence of communication graphs, on which the algorithm needs to solve a problem despite the changes in the communication topology.
Adapting the ideas of smoothed analysis in this setting is not an easy task.
First, while the classical smoothed analysis is applied to a single input, in the dynamic setting we deal with an infinite sequence of communication rounds, on ever-changing graphs.
And second, defining the correct perturbation which the input undergoes is more challenging, as the input is discrete, as opposed to the continuous input of the simplex algorithm, where Gaussian noise is a very natural candidate.


To address the above challenges, Dinitz \etal\ fix a \emph{noise parameter} $k>0$, and then, after the adversary picks an infinite sequence of graphs $\set{G_i}$, 
they derive a new series of graphs $\set{H_i}$ by perturbing every $G_i$ with an addition or deletion of roughly $k$ random edges.
Specifically, the perturbation ($k$-smoothing) is done by choosing uniformly at random a graph within Hamming distance at most $k$ of $G_i$
(i.e., a graph that is different from $G_i$ by at most $k$  edges), which also has some desired properties (say, connectivity).
Note that in the typical case $\Omega (k)$ edges are perturbed, as most of the graphs with Hamming distance at most $k$ from $G_i$ have Hamming distance $\Omega (k)$ from it.
Using this model, they analyze the problems of flooding, hitting time, and aggregation.

In this paper we present a study of \emph{models of smoothing}, or put differently, a study of models of noise in dynamic networks.
To this end, we focus on one of the three fundamental problems presented in the previous work --- the flooding problem.
In this problem, a source vertex has a message it wishes to disseminate to all other vertices in the network. In each round, every vertex which has the message forewords it to all of its neighbors, until all vertices are informed. First, note that the problem is trivially solvable in static networks, in diameter time. Second, in the dynamic setting, it is easy to construct a sequence of graph where flooding takes $\Omega(n)$ rounds, even if each of them has a small diameter.
It is thus not surprising that adding random links between vertices may accelerate the flooding process, and indeed, it was shown in~\cite{DFGN18} that in $k$-smoothed dynamic networks, the complexity of flooding drops to $\tilde{\Theta}(n^{2/3}/k^{1/3})$ (where the $\tilde{\cdot}$ sign hides $\polylog n$ factors). Note the huge improvement, from $\Omega(n)$ to $\tilde{O}(n^{2/3})$, even for $k=1$.


\subsection{Motivation}
The significant step made by Dinitz \etal in introducing smoothed analysis to the distributed dynamic environment left many fascinating questions unanswered. One natural direction of research has to do with applying smoothed analysis for a variety of problems, but here, we take a different approach. In an attempt to understand the essential properties of smoothed analysis in our setting, we focus on questions related to the very basic concept of smoothing and study several models of noise. We outline some of the related questions below.

\subparagraph{The curious gap between $k=0$ and $k=1$} Dinitz \etal show a tight bound of $\tilde{\Theta}(n^{2/3} / k^{1/3})$ for the problem of flooding in a dynamic networks with noise $k$. That is, as $k$ decreases, flooding becomes harder, but only up to $\tilde{\Theta}(n^{2/3})$ for a constant $k$. However, when there is no noise at all, a \emph{linear}, in $n$, number of rounds is required. That is, there is a curious gap between just a tiny amount of noise and no noise at all, which stands in sharp contrast with the smooth change in the flooding time as a function of $k$, when $k$ ranges from $1$ and to $\Omega(n)$. One may wonder if there is a natural way to extend the model presented by Dinitz \etal to bridge this gap.

\subparagraph{An oblivious adversary} The results of Dinitz \etal assume an oblivious adversary. That is, the evolution of the network must be decided by the adversary in advance, without taking into account the noise added by the smoothing process. It is natural to ask what is the power of an \emph{adaptive} adversary compared to an oblivious one. 
As smoothed analysis aims to be the middle ground between worst-case and average-case analysis, it is very natural to consider the effect of noise on the strongest possible adversary, i.e., an adaptive adversary.

\subparagraph{Change-dependent noise} The type of noise considered by Dinitz \etal remains constant between rounds, regardless of the topology or the actions of the adversary, which corresponds to ``background noise'' in the network. However, it is very natural to consider cases where the amount of noise itself changes over time.
More specifically, if we think of a noise as an artifact of the changes the adversary performs to the system, and thus, it is natural to expect different amounts of noise if the adversary performs a single change, or changes the whole topology of the graph.
Hence, a natural model of noise in a network is one where the added noise is \emph{proportional} to the amount of change attempted by the adversary, that is, $k$ is proportional to the Hamming distance between the current graph and the next one. A different, natural definition of a changes-dependent noise, is where the noise occurs only in the changing edges, and not all the graph.


\subsection{Our results}
\begin{table}[tb]
    \centering
    \begin{tabular*}{\linewidth}{@{}l@{\;\;\;}l@{\;\;}l@{\;\; }l@{}}
    \toprule
        Model & Upper bound  & Lower bound & Reference \\
        \toprule

        \begin{tabular}{@{}l@{}}
        Integer noise\\ Non-adaptive adv.
        \end{tabular}
        & $O\left(n^{2/3} (1 / k)^{1/3} \log n \right)$ & $\Omega\left(n^{2/3}/k^{1/3}\right)$, for $k \leq \sqrt{n}$
        & Dinitz \etal~\cite{DFGN18} \\
        \midrule

        \begin{tabular}{@{}l@{}}
        Fractional noise\\ Non-adaptive adv.
        \end{tabular}
        & $O\left(\min\set{n, n^{2/3} (\log n / k)^{1/3}}\right)$ & $\Omega\left(\min\set{n,n/k, n^{2/3}/k^{1/3}}\right)$
        & Thm.~\ref{thm: fractional UB},~\ref{thm: fractional LB 1},~\ref{thm: fractional LB 2}\\
        \midrule

        \begin{tabular}{@{}l@{}}
        Fractional noise\\ Adaptive adv.
        \end{tabular}
        & $O\left(\min\set{n, n \sqrt{\log n / k}}\right)$ & $\Omega(\min \set{n, n \log k / k})$
        & Thm.~\ref{thm: fractional adaptive UB},~\ref{thm: fractional adaptive LB}\\
        \midrule

        \begin{tabular}{@{}l@{}}
        Proportional noise\\ Non-adaptive adv.
        \end{tabular}
        & $O\left(n^{2/3} \cdot ( D \log n / \eps )^{1/3}\right)$ &
        & Thm.~\ref{thm: FP - UB}\\
        \midrule

        \begin{tabular}{@{}l@{}}
        Proportional noise\\ Adaptive adv.
        \end{tabular}
        & $O(n)$ & $\Omega(n)$
        & Thm.~\ref{thm: FP - adaptive LB}\\
        \hline
        Targeted noise
        & $O\left(\min\set{n, D \log n /\eps^{D^2}}\right)$
        &
        $\Omega(n)$, for $D\in\Theta(\log n)$
        & Thm~\ref{thm: flooding targeted UB},~\ref{thm: flooding targeted LB}\\
        \bottomrule
    \end{tabular*}
    \caption{Bounds on flooding time in different models of smoothing}
    \label{table:results}
\end{table}
To address the above points of interest,
we prove upper and lower bounds, as summarized in Table~\ref{table:results}.
First, we show a natural extension of the previously presented noise model to \emph{fractional} values of $k$.
For $k\geq 1$, our model roughly coincides with the prior model and results, while for $0<k<1$, in our model a single random change occurs in the graph with probability $k$.
In our model, we show a bound of $\tilde{\Theta}(\min \{n^{2/3}/k^{1/3},n\})$ for flooding, for values of $k$ ranging from $k=0$ to $k=\Theta(n)$, even fractional --- see theorems~\ref{thm: fractional UB}, \ref{thm: fractional LB 1}, and~\ref{thm: fractional LB 2}.
The flooding time thus has a clean and continuous behavior, even for the range $0 < k < 1$ which was not studied beforehand, providing a very natural extension of the previous model.

Next, we focus our attention on an adaptive adversary, that can choose the next graph depending on the smoothing occurred in the present one.
Here, we show that the added power indeed makes the adversary stronger, and yet, flooding can be solved in this case faster than in a setting where no smoothing is applied.
Specifically, in theorems~\ref{thm: fractional adaptive UB} and \ref{thm: fractional adaptive LB} we show that in this setting flooding can be done in $\tilde{O}(n/\sqrt{k})$ rounds, and takes at least $\tilde{\Omega}(n / k)$ rounds.
This result presents a strict separation between the power of an adaptive and an oblivious adversaries.

We then turn our attention to a different type of noise, and introduce two new models of \emph{responsive noise} --- noise which depends on the changes the adversary preforms.
The goal of studying responsive noise is to better understand systems where the noise is an artifact of the changes, and more changes induce more noise.
We consider two, completely different cases:
if only the \emph{amount} of noise depends on the changes, then the system is less stable than in the prior models, and flooding can be delayed a lot by the adversary.
On the other hand, if the same amount of noise is \emph{targeted} at the changing links themselves, then the system remains stable, and flooding occurs much faster.
In both models, our results surprisingly depend
on a new attribute --- the static diameter of the graph, $D$.

The first model of responsive noise we introduce is the \emph{proportional noise} model, where the noise is randomly spread among the graph edges as before, and its amount is proportional to the number of proposed changes.
We consider two types of adversaries in this setting --- adaptive, and oblivious.
Theorem~\ref{thm: FP - UB} shows that $\tilde{O}(\min \{n^{2/3}D^{1/3}/\eps^{1/3},n\})$ rounds are sufficient for flooding in this model with an oblivious adversary.
Here, the static diameter $D$ comes into play, since the adversary can now force change-free rounds: if the adversary does not make any change, no noise occurs, the graph remains intact and no random ``shortcut edges'' are created.
Current lower bounds for flooding with oblivious adversaries seem to require many changes, but in the proportional noise model this results in the addition of many random edges, thus speeding up the flooding process.
In addition, the upper bound suggests that the static diameter $D$ should come into play in the lower bounds, which is not the case in previous constructions.
While we believe our upper bound is tight, proving lower bounds appears to be beyond the reach of current techniques.

In the proportional noise model with adaptive adversary, we encounter another interesting phenomenon.
Adjusting the upper bound analysis in a straightforward manner gives the trivial upper bound of $O(n)$, and for a good reason: an adaptive adversary can slow down the flooding time all the way to $\Omega(n)$, as shown in Theorem~\ref{thm: FP - adaptive LB}.
The adversary for this lower bound makes only a few changes in each round, and only in necessary spots (using its adaptivity). The fact that the noise is proportional boils down to almost no noise occurring, which allows the adversary to maintain control over the network.

The second model of responsive noise we study is that of a \emph{targeted noise}.
Here, the expected amount of noise applied is roughly the same as above, but only the edges that undergo changes are perturbed by the noise. More concretely, each change proposed by the adversary does not happen with some probability $\epsilon$.
The aim here is to model networks where every attempted change may have some probability of failure.
In this setting, the case of an adaptive adversary seems more natural; however, we prove our upper bound for an adaptive adversary, and the lower bound for an oblivious one, thus handling the harder cases of both sides.

Our upper bound shows significant speedup in flooding on graphs with constant static diameter --- $O(\log n)$ rounds suffice for flooding with targeted noise.
This phenomenon holds even for higher static diameters:
Theorem~\ref{thm: flooding targeted UB} gives an upper bound of $O((D\log_{1/\epsilon} n ) / \epsilon^{D^2})$ rounds for flooding.
This improves upon the trivial $O(n)$ whenever $D = O(\sqrt{\log_{1/\epsilon} n})$. Finally, in Theorem~\ref{thm: flooding targeted LB} we show that for larger static diameter, $D = \Theta(\log n)$, the number of rounds required for flooding is~$\Omega(n)$.


\subparagraph{Our techniques}
Analysing our new models of smoothing require a new and more global techniques. While we adopt the results of \cite{DFGN18} for changes in a single round, our models introduce new technical challenges, as they require multiple-round based analysis.

The main technical challenge for proving upper bounds comes from the fact that one cannot even guarantee that noise occurs at every round, and when it does --- the amount of noise is not fixed through the execution of the whole algorithm.
This requires us to conduct a much more global analysis, taking into account sequences of rounds with a varying amount of noise, instead of analyzing a single round at a time.
In several cases, the static diameter~$D$ appears as a parameter in our results and analysis: in our model, the adversary can force change-free rounds where no noise occurs, in which case flooding happens on a static graph. 


In an attempt to study the exact limitations of our noise models, we present specifically crafted lower bound networks for each model.
Note that in many models our adversary is adaptive, potentially making it much stronger. This makes our lower bound more strategy-based, and not merely a fixed instance of a dynamic graph.
We revise the original flooding lower bound of Dinitz \etal, in order to make it more applicable to other models. We present a more detailed and rigorous proof, that indeed implies tight bounds in most of our models, and specifically, when considering adaptive adversaries.

\hide{Finally, we also show some results that combine the the type of noise defined in \cite{DFGN18} with our definition of responsive noise. Intuitively, one can see the noise defined in \cite{DFGN18} as a sort of noise which is always present in the network, a kind of \emph{background noise}.
We show that as long as no round permits more than a linear number of changes made by the adversary, the background noise \emph{overcomes} the responsive one. }

\subsection{Related work}
Smoothed analysis was introduced by Spielman and Teng~\cite{SpielmanT09,SpielmanT04} in relation to using pivoting rules in the simplex algorithm. Since, it have received much attention in sequential algorithm design; see, e.g., the survey in~\cite{SpielmanT09}. The first application of smoothed analysis to the distributed setting is due to Dinitz \etal~\cite{DFGN18}, who apply it to the well studied problems of aggregation, flooding and hitting time in dynamic networks, as described above.
Chatterjee \etal~\cite{abs-1911-02628} considered the problem of a minimum spanning tree construction in a distributed setting of a synchronous network, where the smoothing is in the form of random links added in each round, which can be used for communication but not as a part of the tree.

The field of distributed algorithm design for dynamic networks has received considerable attention in recent years~\cite{Michail16}.
Most of the works considered models of adversarial changes~\cite{KuhnLO10,AugustinePRU12, DuttaPRSV13,GhaffariLN13, HaeuplerK11, KuhnOM11,BambergerKM19, DCKPS2020},
while others considered random changes, or random interactions between the vertices~\cite{CaiSZ20,ClementiST15,DenysyukR14,GasieniecS18,KowalskiM2020}.
Yet, as far as we are aware, only the two aforementioned works take the smoothed analysis approach in this context.

\section{Preliminaries}

\subsection{Dynamic graphs}
All graphs are of the form $G=(V,E)$, where $V=[n]=\{1,\ldots,n\}$ is the set of vertices.
A dynamic graph $\Hcal$ is a sequence $\Hcal=(G_1,G_2,\ldots)$ of graphs, $G_i=(V,E_i)$ on the same set of vertices, which can be finite or infinite.
The \emph{diameter} of a (static) graph is the largest distance between a pair of vertices in it, and the \emph{static diameter} of a dynamic graph $\Hcal$ is the minimal upper bound $D$ on all the diameters of the graphs in the sequence $\Hcal$.

We study synchronous distributed algorithms in this setting, where the time is split into rounds, and in  round $i$ the vertices communicate over the edges of the graph $G_i$.

Given two graphs $G=(V,E)$ and $G'=(V,E')$ on the same set of vertices, we denote $G-G'=(V,E\setminus E')$.
We also denote $G\oplus G'=(V,E\oplus E')$, where $E\oplus E'$ is the set of edges appearing in exactly one of the graphs.
The size of set $E\oplus E'$ is called the \emph{Hamming distance between $G$ and $G'$}, and is denoted by $|G\oplus G'|$.

\subsection{Models of noise}
Our smoothed analysis is based on three models of noise, defined in this section.

For a single step, we recall the definition of $t$-smoothing~\cite{DFGN18}, and add the new notion of $\epsilon$-smoothing, for a parameter $0 < \epsilon < 1$.
At each step, we think of $\gold$ as the current graph, and of $\gtmp$ as the future graph suggested by the adversary.
The actual future graph, $\gnew$, will be a modified version of $\gtmp$, randomly chosen as a function of $\gold$ and $\gtmp$.
In addition, we consider the set $\gallowed$ of \emph{allowed graphs} for a specific problem.
For flooding, this is just the set of all connected graphs.

\hide{
We start with two auxiliary definition. Let $0<\epsilon<1$ and $k\in\bbR_{+}$ be two parameters and $\gold$, $\gtmp$ two graphs.
We think of $\gold$ as the current graph, and of $\gtmp$ as the future graph suggested by the adversary.
The actual future graph, $\gnew$, will be a modified version of $\gtmp$, randomly chosen as a function of $\gold$ and $\gtmp$.
In addition, we consider the set $\gallowed$ of \emph{allowed graphs} for a specific problem.
For flooding, this is just the set of all connected graphs.
}

\hide{
We start by defining a single smoothing step, extending the definition of $k$-smoothing from~\cite{DFGN18} to allow also targeted noise.
}

\begin{definition}
Let $0\leq\epsilon\leq 1$ and $t \in \mathbb N$ be two parameters, $\gallowed$ a family of graphs, and $\gold$ and $\gtmp$  two graphs in $\gallowed$.
\begin{itemize}
	\item
	A $t$-smoothing of $\gtmp$ 
	is a graph $\gnew$ which is picked uniformly at random from all the graphs of $\gallowed$ that are at Hamming distance at most $t$ from $\gtmp$.
	The parameter $t$ is called the \emph{noise parameter.}
	
	\item
	An $\epsilon$-targeted smoothing of a graph $\gtmp$ with respect to $\gold$ is a graph $\gnew$ which is constructed from $\gtmp$ by adding to it each edge of $\gold-\gtmp$ independently at random with probability $\epsilon$, and removing each edge of $\gtmp-\gold$ with the same probability.
	If the created graph is not in $\gallowed$, the process is repeated.
\end{itemize}	
\end{definition}


We are now ready to define the three types of smoothing for dynamic graphs. The first extends the definition of~\cite{DFGN18} for non-integer values $k\in\bbR_{+}$, and the other two incorporate noise that depends on the latest modifications in the graph.

For a positive real number $x$, we define the random variable $\roundp(x)$, which takes the value $\ceil{x}$ with probability $x - \floor{x}$ (the fractional part of $x$), and $\floor{x}$ otherwise.

\begin{definition}
\label{def: smoothed_dynamic_network}
Let $0\leq\epsilon\leq 1$ and $k\in\bbR_{+}$ be two parameters, and $\gallowed$ a family of graphs.

Let $\Hcal=(G_1,G_2,\ldots)$ be a dynamic graph, i.e., sequence of ``temporary'' graphs. Let $G'_0$ be a graph (if $G'_0$ is not explicitly specified, we assume $G'_0=G_1$).
\begin{itemize}
	\item
	A $k$-smoothed dynamic graph is the dynamic graph $\Hcal'=(G'_0,G'_1,G'_2,\ldots)$ defined from $\Hcal$, where 
	for each round $i>0$,  $G'_i$ is the $t_i$-smoothing of $G_i$, where $t_i=\roundp(k)$.

	\item
    An $\eps$-\emph{proportionally} smoothed dynamic graph $\Hcal'=(G'_0,G'_1,G'_2,\ldots)$ is defined from $\Hcal$, where 
    for each round $i>0$,  $G'_i$ is the $t_i$-smoothing of $G_i$, 
    where $t_i=\roundp(\epsilon \cdot \size{G'_{i-1}\oplus G_{i}})$.
	
	\item
	An $\epsilon$-\emph{targeted} smoothed dynamic graph is the dynamic graph $\Hcal'=(G'_0,G'_1,G'_2,\ldots)$ defined iteratively from $\Hcal$, where 
	for each round $i>0$,  $G'_i$ is the $\epsilon$-targeted smoothing of $G_i$ w.r.t.~$G'_{i-1}$.
\end{itemize}

All the above definitions also have \emph{adaptive} versions, where each $\Hcal$ and $\Hcal'$ are defined in parallel: the graph $G_i$ is chosen by an adversary after the graphs $G_0,G_1,G_2,\ldots,G_{i-1}$ and $G'_0,G'_1,G'_2,\ldots,G'_{i-1}$ are already chosen, and then $G'_i$ is defined as above.
We use adaptive versions for the first two definitions.
\end{definition}

\subsection{Auxiliary lemmas}
We use two basic lemmas, which help analyzing the probability of the noisy to add or remove certain edges from a given set of potential edges.
These lemmas address a single round, we apply them with a different parameter in each round.
The lemmas appeared as Lemmas~4.1 and 4.3 in~\cite{DFGN18},
where they were used with the same parameter throughout the process.



\begin{lemma}
\label{lem: LB hitting a set of edges}
    There exists a constant $c_1>0$ such that the following holds.
    Let $\gtmp \in \gallowed$ be a graph, and $\emptyset \neq S \subseteq \binom{[n]}{2}$ a set of potential edges.
    Let $t\in\bbN$ be an integer such that $t\leq n/16$ and $t\cdot \size{S} \leq n^2/2$.

    If $\gnew$ is a $t$-smoothing of $\gtmp$, then the probability that $\gnew$ contains at least one edge from $S$ is at least $c_1 \cdot t \size{S} / n^2$.
\end{lemma}




\begin{lemma}
\label{lem: UB adding from a set of edges}
    There exists a constant $c_2>0$ such that the following holds.

    Let $\gtmp \in \gallowed$ be a graph. Let $S \subseteq \etmp$ be a set of potential edges  such that $S \bigcap \etmp = \emptyset$. Let $t\in\bbN$, such that $t\leq n/16$.

     If $\gnew$ is an $t$-smoothing of $\gtmp$, then the probability that $\gnew$ contains at least one edge from $S$ is at most  $c_2 \cdot t \size{S} / n^2$.
\end{lemma}

\subsection{Probability basics}
In all our upper bound theorems, we show that flooding succeeds with high probability (w.h.p.), i.e., with probability at least $1-n^{-c}$ for a constant $c$ of choice.
For simplicity, we prove success probability $1-n^{-1}$, but this can be amplified by increasing the flooding time by a multiplicative constant factor.
Our lower bounds show that flooding does not succeed with probability more than $1/2$, so it definitely does not succeeds w.h.p.

We will also use the following Chernoff bound (see, e.g.~\cite{MitzenmacherU05}):
\begin{lemma}
    Suppose $X_1,\dots,X_n$ are independent Bernoulli random variables, with $\Expc{X_i} = \mu$ for every $i$.
    Then for any $0 \leq \delta \leq 1$:
    \[\Pr\left[ \sum_{i=1}^{n} X_i \leq (1-\delta)\mu n \right] \leq e^{ - \delta^2 \mu n / 2}.\]
\end{lemma}
We usually apply this bound with $\delta = 0.9$.

\hide{ 
We add another lemma, to show that sometimes that chance of not hitting a set $S_1$, but hitting a different set $S_2$ can relate to the sizes of said sizes.

\begin{lemma}[Lemma 4.3 in~\cite{DFGN18}]
\label{lem: UB adding from a set of edges}
    There exists a constant $c_4>0$ such that the following holds.
    Let $\gold,\gtmp$ be two graphs on the same set $[n]=\{1,\ldots,n\}$ of vertices and $\emptyset\neq S\subseteq \binom{[n]}{2}$ a set of potential edges such that $S \bigcap \etmp = \emptyset$.
    Let  $0<\epsilon<1$ and $k$ be parameters such that for $t=\roundp(\max \set{k, \epsilon \cdot \size{\gold-\gtmp}})$ we have $t\leq n/16$.

    If $\gnew$ is an $(\eps,k)$-smoothing of $\gtmp$ w.r.t.~$\gold$, then the probability that $\gnew$ containing at least one edge from $S_2$, and no edges from $S_1$ is at most:  $c_2 t \size{S} / n^2$.
\end{lemma}

\begin{proof}
    We denote by $E_2$ the event that $\gnew$ contains at least one edge of $S_2$ and by $E_1$ the event of $\gnew$ having no edges from $S_1$. We also focus on a specific spanning tree $T$ of size $n-1$ edges, and denote by $E_T$ the probability of the noise not hitting $T$ is at least $1/2$. For $t < n/16$ noise, we know that $\Pr[E_T] \geq 1/2$ (see original proof of Lemma~\ref{lem: LB hitting a set of edges}), and therefore
    $$\Pr[E_1 \wedge E_2 | E_T] \leq \frac{\Pr[E_1 \wedge E_2 \wedge E_t]}{\Pr[E_T]} \leq 2\Pr[E_1 \wedge E_2]$$

\end{proof}
}

\section{Flooding in Dynamic Networks}
\label{sec: background noise}

\hide{ 
\begin{theorem}
For an $(\epsilon,k)$-smoothed dynamic graph, flooding takes $\Theta(n^{2/3} / k^{1/3})$ rounds, when $k\geq 1$.
\end{theorem}
\begin{proof}
The upper bound from~\cite{DFGN18} holds, as in each round we have noise of magnitude \emph{at least} $k$. Upper bound from Seth's paper still holds. Lower bound: take spooling graph, but instead of connecting center to center, connect center to a vertex on the edge of the star. Now the number of changes is constant. Also we clean any edges between the two starts incurred due to noise. This requires more actions (due to multiplicative noise) but this sum converges and is bounded by $3k$. The rest of the proof should be similar to the original.
\end{proof}
}

\subsection{Fractional amount of noise}
We address an interesting phenomenon which occurs in~\cite{DFGN18}: by merely adding a single noisy edge per round, the flooding process dramatically speeds up from $\Theta(n)$ in the worst case, to only $\Theta(n^{2/3})$ in the worst case.
To explain this gap, we present a natural notion of \emph{fractional} noise: if the amount of noise $k$ is not an integer, we use an integer close to $k$, using the function~$\roundp(k)$.
An easy assertion of the result in \cite{DFGN18}, is that whenever $k > 1$, the analysis does not change by much, and the same asymptotic bound holds. This leaves open the big gap between $k=0$ (no noise) and $k=1$. We address this question in the next two theorems, showing a smooth transition in the worst-case flooding time between these two values of $k$.

Next, we revise the upper bound in~\cite{DFGN18} to handling fractional values of $k$ and values smaller than $1$. In addition, we show a somewhat cleaner argument, that carries three additional advantages: (1) It can be easily extended to adaptive adversaries; (2) It extends well to other models, as seen in subsequent sections; (3) It decreases the multiplicative logarithmic factor, dropping the asymptotic gap between upper and lower bound to only $O(\log^{1/3}(n))$.
Hence, the proof of the next theorem serves as a prototype for later proofs, of upper bounds for adaptive adversaries and proportional noise.

\begin{theorem}
    \label{thm: fractional UB}
    Fix a real number $0 < k \leq n/16$. For any $k$-smoothed dynamic graph, flooding takes $O\left(\min\set{n, n^{2/3} (\log n / k)^{1/3}}\right)$ rounds, w.h.p.
\end{theorem}

Note that an $n$-round upper bound is trivial (using the connectivity guarantee), and that whenever $k \leq c \cdot \log n / n$ for some constant $c$, the term $n$ is the smaller one, and we need not prove anything more.%
\footnote{This is rather intuitive, as if $k$ is very small (e.g., $k \leq c/n$), when considering $\delta n$ rounds for small enough $\delta$, no noise occurs. It is known that when no noise occurs, a linear amount of rounds is required for flooding. This intuition is formalized in the lower bound later in this section.}
We therefore focus on the event of $k > c \cdot \log n / n$, and for this case we show an upper bound of $O(n^{2/3}(\log n / k)^{1/3})$.

\begin{proof}
    Fix $u$ to be the starting vertex, and $v$ to be some other vertex. We show that within $R = 3r$ rounds, the message has been flooded to $v$ with constant probability (over the random noise). We split the $3r$ rounds into three phases as follows.
    \begin{enumerate}
        \item First $r$ rounds, for spreading the message to a large set of vertices we call $U$.
        \item Next $r$ rounds, for transferring the message from the set $U$ to another large set $V$ to be defined next.
        \item Last $r$ rounds, for ``converging'' and passing the message from the potential set $V$ to the specific vertex $v$.
    \end{enumerate}
    We now quantify the message flooding in the three phases.

    \subparagraph{After the first $r$ rounds, the message has reached many vertices.}
    Let $U_i$ denote the set of vertices that have been flooded at time $i$. Using merely the fact that the graph is connected at all times, and assuming the flooding process is not over, we know that at each round a new vertex receives the message, implying $\size{U_{i+1}} \geq \size{U_i} + 1$.
    Hence, $\size{U_r} \geq r$ (with probability~$1$).

    \subparagraph{In the last $r$ rounds, many vertices can potentially flood the message to $v$.}
    We denote by $V_i$ the set of vertices from which $v$ can receive the message within the last $i$ rounds ($R-i+1, \dots, R$).
    Formally, fix the sequence of graphs chosen by the oblivious adversary,
    and apply the random noise on these graphs to attain the last $r$ graphs of our \emph{smoothed} dynamic graph.
    Define $V_0=\set{v}$, and for $i>0$, define $V_i$ as the union of $V_{i-1}$ with the neighbors of it in $G_{R-i+1}$.
    That is, $V_i$ is defined analogously to $U_i$ but with the opposite order of graphs.
    We point out that we are dealing with a stochastic process: each $V_i$ is a random variable that depends on the noise in the last $i$ rounds. Still, the connectivity guarantees that $\size{V_{i+1}} \geq \size{V_{i}} + 1$. Therefore, we have $\size{V_{r}} \geq r$ (with probability $1$).%
    \footnote{Note that here we strongly rely on the obliviousness of the adversary: with an adaptive adversary, one cannot define $V_r$ properly before the end of the execution of the algorithm, as the adversary's choices were not yet made. The case of an adaptive adversary is discussed in the next section.}

    \subparagraph{The middle rounds.}
    The above processes define two randomly chosen sets, $U_{r}$ and $V_{r}$, each of size at least $r$. If $U_{r} \cap V_{r} \neq \emptyset$, then we are done as $v$ is guaranteed to receive the message even if we ignore the middle rounds.
    Otherwise, we consider the set $S = U_{r} \times V_{r}$ of potential edges, and show that one of them belongs to our smoothed dynamic graph in at least one of the middle graphs, with probability at least $1 - n^{-2}$, for the right value of $r$.

    Let us consider separately the two cases: $k \geq 1$ (note that non-integer $k$ was not handled before), and the case $0 < k <1$.
    \subparagraph{The case of $k \geq 1$.}
    In this case, we essentially claim the following: we are guaranteed to have either $\floor{k}$ noise or $\ceil{k}$ at each and every round. Applying Lemma~\ref{lem: LB hitting a set of edges} for each such round, the probability of not adding any edge from $S$ is at most
    \[\left(1 - c_1 \floor{k} \size{S} / n^2 \right) \leq \left(1 - c_1 k r^2 / 2n^2 \right),\]
    where the inequality follows from $\floor{k} \geq k/2$ and $\size{S} = \size{U_{r}} \cdot \size{V_{r}} \geq r^2$.
    Thus, the probability of not adding any edge from $S$ in any of these $r$ noisy rounds is at most
    \[\left(1 - c_1 k r^2 / 2n^2 \right)^{r} \leq e^{-c_1 r^3 k /(2n^2)},\]
    which is upper bounded by $n^{-2}$, whenever $r \geq n^{2/3} \cdot \left[(4 \log n)/(c_1 k)\right]^{1/3}$.

    \subparagraph{The case of $0 < k < 1$.}
    Note that here we can no longer use $\floor{k} \geq k/2$, and so we turn to a somewhat more global analysis which guarantees that ``enough'' rounds produce a (single) noisy edge:
    at each round we essentially add a noisy edge with probability $k < 1$, and otherwise do nothing. Since we have $r$ rounds, and in each of them noise occurs  independently, we can use a standard Chernoff bound to say that with all but probability at most $e^{-0.4kr}$, we have $(k/10)r$ rounds in which a noisy edge was added.\footnote{The expectation over all rounds is $kr$ noisy edges.}
    We later bound this term. For each of the $(k/10)r$ rounds we again apply Lemma~\ref{lem: LB hitting a set of edges} to claim that no edge from $S$ was added, with probability at most
    \[\left(1 - c_1 \size{S} / n^2 \right) \leq \left(1 - c_1 r^2 / n^2 \right),\]
    where the inequality follows again from $\size{S} = \size{U_{r}} \cdot \size{V_{r}} \geq r^2$.
    Thus, the probability of not adding any edge from $S$ in any of these $(k/10)r$ noisy rounds is upper bounded by
    \[\left(1 - c_1 r^2 / n^2 \right)^{(k/10) r} \leq e^{-c_1 k r^3 /(10n^2)},\]
    which is upper bounded by $1/(2n^2)$ whenever $r \geq n^{2/3} \cdot \left[(10\log n)/(c_1 k)\right]^{1/3}$.

    Recalling that we only deal with the case $k \geq c\cdot \log n / n$, choosing the constant $c$ according to the constant $c_1$, we also get $r \geq 10\log n / k$. This allows us to bound the error of the Chernoff inequality: $e^{-0.4kr} \leq e^{-3 \log n} \leq 1/(2n^2)$. Union bounding on both possible errors, we know that with probability at least $1-n^{-2}$ the vertex $v$ is informed after $3r$ rounds.


    \medskip

    To conclude, we use $r = n^{2/3} \cdot \left[(10\log n)/(c_1 k)\right]^{1/3}$. For any value $k \geq c \cdot \log n / n$, and for any realization of $U_{r}$ and $V_{r}$, the message has been passed to $V_{r}$ within the $r$ middle rounds, with probability at least $1-n^{-2}$. So $v$ received the message after $R = 3r$ rounds with the same probability.
    Taking a union bound over all the vertices implies that $R$ rounds are enough to flood to the whole network with probability at least $1-1/n$.
\end{proof}

We also show a matching lower bound, restructuring the proof of the lower bound in~\cite{DFGN18}. We note that their proof actually states a lower bound of $\Omega\left(min\set{n/k, n^{2/3}/k^{1/3}}\right)$ that apply to any $k < n/16$ (and is indeed tight for $k < O(\sqrt{n})$). We restructure the analysis of the lower bound to fit the different constructions given in the following sections, as follows.

First, we consider the first $R$ rounds, for some parameter $R$, and show inductively that the following two main claims continue to hold for every round $i < R$ with very high probability.
\begin{itemize}
    \item Some pre-defined set of vertices stays uninformed
    \item Given the above, the \emph{expected} number of informed vertices in total does not grow rapidly.
\end{itemize}
The growth (in the expected number of informed vertices) has a multiplicative-additive progression, potentially becoming an exponential growth (with a small exponential).

We choose $R$ such that the expected number of informed vertices is upper bounded by $\delta n$, and use Markov inequality to show that with high probability the flooding is not yet over after $R$ rounds. We then  apply a union bound over all the inductive steps, where each has a small probability to fail the inductive argument. Altogether, we show that with high probability, the flooding is not over after $R$ rounds.

In Appendix~\ref{Sec: appendix thm proof}, we use an argument along those lines in order to prove the following extensions of the lower bound from~\cite{DFGN18} to non-integer values and to fractional values (where $k < 1$).

\begin{theorem}
    \label{thm: fractional LB 1}
    Fix $1 \leq k \leq n/16$ (not necessarily an integer). For any $k$-smoothed dynamic graph, for the flooding process to succeed with probability at least $1/2$, it must run for $\Omega\left(\min\set{n/k, n^{2/3}/k^{1/3}}\right)$ rounds.
\end{theorem}

In particular, whenever $k = O(\sqrt{n})$, the dominant term is $\Omega(n^{2/3}/k^{1/3})$, matching the upper bound (up to logarithmic factors).

\begin{theorem}
    \label{thm: fractional LB 2}
    Fix $0 < k < 1$. For any $k$-smoothed dynamic graph, for the flooding process to succeed with probability at least $1/2$, it must run for $\Omega\left(\min\set{n, n^{2/3}/k^{1/3}}\right)$ rounds.
\end{theorem}

\subsection{Adaptive vs. oblivious adversary}
Here we note that the results of \cite{DFGN18} are for the case of oblivious (non-adaptive) adversary.
We extend our results (in the more generalized, fractional noise regime) to the adaptive case, obtaining bounds that are different than the ones in \cite{DFGN18}.
Interestingly, our results show separation between adaptive and oblivious adversary for a wide range of network noise: for constant $k$, in particular, we get a separation for the adaptive case, where no constant amount of noise speeds up the flooding below $\Omega(n)$, and the oblivious case where $k = \omega(1/n)$ already speeds the flooding to $o(n)$ rounds.

\begin{theorem}
    \label{thm: fractional adaptive UB}
    Fix $0 < k \leq n/16$, not necessarily an integer. For any $k$-smoothed \emph{adaptive} dynamic graph, flooding takes $O(\min\set{n, n \sqrt{\log n / k}})$ rounds, w.h.p.
\end{theorem}

    The proof follows an argument similar to the one of Theorem~\ref{thm: fractional UB}, where the process is broken to three phases: $u$ spreads the message to $U_r$, which carry it to $V_r$, who are sure to deliver it to the destination $v$. The major difference here is that we can no longer rely on the last phase, since an adaptive adversary \emph{sees} the informed vertices and is able to act accordingly. We thus give more focus to the second, analyzing the chance for direct noisy edge between the informed set $U_r$ and a destination vertex $v$.

\begin{proof}
    First, note that either way $n-1$ rounds are enough for flooding, as we still have a connectivity guarantee for each and every iteration. This means that for $k < \Theta(\log n)$, the theorem does not improve upon the trivial upper bound. In particular, in the following, we can therefore disregard the case of $k < 1$.

    We use a similar argument as in the proof of Theorem~\ref{thm: fractional UB}.
    We note that the previous proof relied on an oblivious adversary when we statistically analysed the middle rounds, for each \emph{fixed} value of $V_r$ (which depends on the the last $r$ rounds).
    Assuming adaptivity, one can no longer use this analysis, and so we turn to a more simplified analysis: we consider only the two phases: the first $r$ rounds and $r$ last rounds. At first, connectivity assures us that enough vertices learn the message, and next we analyse the probability of one of them to connect directly to the goal vertex $v$ by a noisy edge. As before, the first phase gives us $U_r \geq r$ with probability $1$. Next, we give analysis of the second phase, that correspond to the one of the middle rounds before.
    \subparagraph{The second phase}
    Note that if $v \in U_{r}$, then we are done. Otherwise, we consider the set $S = U_{r} \times \set{v}$ of potential edges, and show that one of them belongs to our smoothed dynamic graph in at least one of the intermediate graphs, with high probability. We use the fact that for $0 < k < 1$ the assertion is trivial, and analyse the case of $k \geq 1$.

    In this case, as before, we apply Lemma~\ref{lem: LB hitting a set of edges} for each round in this phase, and conclude that the probability of not adding any edge from $S$ is at most
    \[1 - c_1 \floor{k} \size{S} / n^2  \leq
    1 - c_1 k r / 2n^2 ,\]
    where the inequality follows from $\floor{k} \geq k/2$ and $\size{S} = \size{U_{r}}$.
    Thus, the probability of not adding any edge from $S$ in any of these $r$ noisy rounds is upper bounded by
    \[\left(1 - c_1 k r / 2n^2 \right)^{r} \leq e^{-c_1 r^2 k /(2n^2)},\]
    which is upper bounded by $n^{-2}$, for $r = 2n \cdot \sqrt{\log n/(c_1 k)}$.

    \medskip

    Using this value for $r$, for any value $k \geq 1$, and realization of $U_{r}$, the message has been passed to $v$ with probability at least $1 - n^{-2}$, within the whole $R=2r$ rounds of phase one and two. A  union bound over all vertices $v \neq u$ implies that $R$ rounds are enough to flood to the whole network with probability at least $1-1/n$.
\end{proof}

When $k$ is a constant, the above result is tight, which we prove by adjusting the oblivious dynamic network from Theorem~\ref{thm: fractional LB 1} to use the adversary's adaptivity, as follows.
The basic structure of the hard instance for flooding is achieved by roughly splitting the graph into \emph{informed} and \emph{uninformed} vertices. In order to ensure low diameter, all the uninformed vertices are connected as a star, centered at a changing uninformed vertex called a \emph{head}, which in turn connects them to a star composed of the informed vertices.
The key idea in the analysis of this process, is that the head at each round should not become informed via noisy edges at an earlier round, as in this case it will immediately propagate the message to all the uninformed vertices, completing the flooding.
An oblivious adversary picks a sequence of heads at the beginning of the process, and this invariant is kept since the probability of any of the selected heads to get informed too early is low. However, after roughly $n^{2/3}$ rounds, the probability that a selected head is informed becomes too high.
An adaptive adversary, on the other hand, can continue crafting a hard instance graph for a linear number of rounds --- in each round, the adversary knows which vertices are uninformed and can pick one of them as the new head, thus overcoming this obstacle.

\begin{theorem}
    \label{thm: fractional adaptive LB}
    Fix $0 < k \leq n/16$, not necessarily an integer. For any $k$-smoothed \emph{adaptive} dynamic graph, for the flooding process to succeed with probability at least $1/2$ it must run for $\Omega(\min \set{n, n \log k / k})$ rounds.
\end{theorem}

\begin{proof}
    We start by formally defining a hard case for flooding with noise.
    The base to our definition is the \emph{spooling graph}, which was defined by Dinitz \etal~\cite{DFGN18}:
    at time $i$, the edges are $\set{(j,i)}_{j=1}^{i-1}$, a star around the vertex~$i$ (this will be the informed spool)
    and edges $\set{(i+1,j)}_{j=i+2}^{n}$, a star around the vertex $i+1$, which will be the right spool, with vertex $i+1$ crowned as the \emph{head}. the spools are then connected using the additional edge $(i, i+1)$.

    We define the \emph{adaptive spooling graph}, as follows: at first $E_1 = ([n]\setminus\set{2}) \times \set{2}$. After the first iteration, $2$ learns the message and is removed to the left (informed) spool. We denote for each round $i$ the set of informed vertices at the end with $I_i$. We also use $u_i$ for the lowest index of an uninformed vertex at the end of round $i$. Formally $u_i = \min\set{\bar{I}_i}$. We define the adaptive spooling graph at time $i+1$ to have the edge set
    \[E_{i+1} =  \set{1} \times (I_i\setminus \set{1}) \ \cup \ \set{(1,u_i)} \ \cup \ \set{u_i} \times (\bar{I}_i \setminus \set{u_i}).\]

    In this graph, essentially, $1$ is connected to all already-informed vertices, and also to the next head $u_i$, which is connected to all other uninformed vertices. The static diameter stays $3$ at each iteration, but now we are promised that at iteration $i+1$, the new head $u_{i}$ is uninformed at the beginning of said round.

    We next show that the \emph{expected} number of informed players cannot grow by much.

    \begin{claim}
    \label{claim: expected informed vertices}
        When taking expectation over the noise we add, we have
        $$\Expc{\size{I_{i+1}}} \leq (1+c_2 k/n)\cdot \Expc{\size{I_{i}}} + 1 .$$
    \end{claim}
    Note that $I_i$ does not depend on the noise of rounds $i+1,\dots,R$, and so the left side takes expectation over $i+1$ rounds of noise, and the right side over $i$ rounds.

    Intuitively, the claim states that the expected growth in the number of \emph{informed} vertices is bounded by an additive-multiplicative progression (i.e., at each step the amount grows multiplicatively and then a constant is added). This is true as the noisy edges induce a multiplicative growth, and connectivity forces that at least $1$ additional vertex (in fact, $u_i$ itself) receives the message. The proof of this is deferred to Appendix~\ref{Sec: appendix thm proof}.

    Using the claim, we analyze the progression of $A_i = \Expc{\size{I_i}}$, splitting to $2$ cases:
    \begin{enumerate}
        \item For $A_i \leq n/(c_2 k)$, the additive term is larger, and so in this range $A_{i+1} \leq A_i + 2$.
        \item For $A_i \geq n/(c_2 k)$, the multiplicative term overcomes and we get $A_{i+1} \leq (1 + 2c_2 k/n)A_i$.
    \end{enumerate}
    Note that we also have $A_{i+1} \geq A_i + 1$, using connectivity, showing this bounds on the progression are rather tight.

    We next split into cases: if $k < 1/c_2$, the additive term controls the process through $\Theta(n)$ iterations, giving us $A_{n/20} \leq n/10$.
    Otherwise, the multiplicative term would come into play.
    We denote the number of additive rounds by $r_0$: the minimal index such that $A_{r_0} \geq n/(c_2 k)$.
    Since $A_0 = 1$ (only the vertex $1$ knows the message at the beginning), and we have good bounds on the progression, we conclude $r_0 = \Theta(n/k)$.

    Next, we allow additional $r_1$ rounds in which the multiplicative term is dominant. We get
    \[A_{r_0 + r_1} \leq A_{r_0} (1+2c_2 k/n)^{r_1} \leq \Theta(n/k) (1+2c_2 k/n)^{r_1}.\]
    Taking $r_1 = \delta(n\log k / k)$ with small enough $\delta$, we have $A_{r_0+r_1} = \Expc{\size{I_{r_0+r_1}}}\leq n/10$.

    In both cases, there is a round~$R$ with $\Expc{\size{I_R}} \leq n/10$.
    By Markov inequality,
    strictly less than $n$ vertices are informed after $R$ rounds, w.p. at least~$0.9$.
    Note that for the first case $R = \Theta(n)$, and for the second case $R  = r_0 + r_1 = O(n\log k /k)$, concluding the proof.
\end{proof}

\section{Responsive noise}
\label{sec: responsive noise}
In this section we consider \emph{responsive} noise in the dynamic network, where some elements of the noise incurred in each round relate to the changes done at this round.
We consider this model as complement to one discussed in the previous section: the parameter $k$ of ``noise per round'' is fit to model some internal background noise in a system.
However, in a more realistic model we would expect the noise to work in different patterns in times of major changes in the network, as opposed to rounds where no changes were made at all.

To this end, we introduce two variants of a responsive noise: one is that of \emph{proportional} noise, and the other is \emph{targeted} noise.
In the first, the amount of noisy edges would relate to the amount of changes that occurred in the last step.
In the second variant, we expect the noise to specifically ``target'' newly modified edges. This could relate to a somewhat limited adversary (in aspect of ability to change the graph).

In the responsive model, an interesting phenomenon occurs: an adversarial network can choose to stay still, in order to force a round where no noise occurs.
For the flooding problem, this strength is limited: whenever the static diameter of each iteration is upper bounded by $D$, the waiting game can only last $D-1$ rounds. To that end, we show how this model affects the analysis of the upper bound, and yet again incorporate the new phenomenon described to devise a non-trivial lower bound.

\subsection{Proportional noise}

\begin{theorem}
    \label{thm: FP - UB}
    Fix $0 < \eps$. For any $\eps$-proportionally smoothed dynamic graph with static diameter at most $D$. If no noise invokes more than $n/16$ changes, the flooding process finishes after $O(n^{2/3} \cdot ( D \log n / \eps )^{1/3})$ rounds, w.h.p.
\end{theorem}

The proof resembles the proof of Theorem~\ref{thm: fractional UB} with three phases, which handles the ``waiting game'' mentioned above. Given static diameter $D$, an adversarial network can only stay intact for $D-1$ rounds, thus in the second phase at least $1/D$ of the rounds introduce one changed edge (or more), inferring that w.h.p. noise is incurred in $\Omega(\eps/D)$ rounds.

\begin{proof}
    We start by noting that if $n \leq c \cdot D \log n / \eps$, for some constant $c$, this bound is larger than the trivial $O(n)$ guaranteed by connectivity and we are done. Next, we follow a similar argument as the one in the proof of Theorem~\ref{thm: fractional UB}.
    For the oblivious case, we had 3 phases: first $r$ rounds and last $r$ rounds are always guaranteed connectivity, and so we have the random sets $U_r, V_r$ such that with probability $1$ we know $\size{U_r}, \size{V_r} \geq r$. If the sets intersects we are done, so assume they do not, and denote $S = U_r \times V_r$. Note that $\size{S} \geq r \cdot r = r^2$.
    Note that the size of $S$ is guaranteed even though its realization depends on the first and last rounds. Since the adversary is not adaptive, we can turn to analyse chance of adding a noisy edge from the set $S$ during the $r$ middle rounds. We note that the adversary can now play the waiting game in order to force a round where no noise is invoked.
    \subparagraph{The middle rounds.}
    As we are promised a static diameter of at most $D$, we know that if the networks does not change for $D$ steps consecutively, the flooding is guaranteed to finish successfully. Therefore, we can assume that within the $r$ middle rounds, once every $D$ rounds, we have a round where changes happen, and therefore potential noise might occur.
    This overall guarantees us $r/D$ such rounds,
    where in each of them noise occurs independently with probability at least $\eps$ (that correspond to the amount of noise if only a single change was made).
    We apply a standard Chernoff bound to say that with all but probability at most $e^{-0.4\eps r / D}$, we have $\eps r/(10D)$ rounds in which at least one noisy edge was added.%
    \footnote{The expectation over all rounds is at least $\eps r / D$ rounds with noisy edges.} 
    We later bound this term.
    \hide
    {
    Using the same argument as in the proof of Theorem~\ref{thm: fractional UB} (for the case $k < 1$) we get by Chernoff bound that with $90\%$, at least $\eps r / (10D)$ rounds invoked noise of at least a single edge.
    }
    We write $t_i$ for the amount of noise at each such round, for which we know $1 \leq t_i \leq n/16$, using the premise. Now, one can safely apply Lemma~\ref{lem: LB hitting a set of edges} with $S$ for each such round, to upper bound the probability of not adding any such potential edge:
    \[\left(1 - c_1 t_i\size{S} / n^2 \right) \leq \left(1 - c_1 r^2 / 2n^2 \right),\]
    where the inequality simply follows from $t_i \geq 1$ and $\size{S} \geq r^2$.
    Thus, the probability of not adding any edge from $S$ in any of the $\eps r/(10D)$ noisy rounds is upper bounded by
    \[\left(1 - c_1 r^2 / 2n^2 \right)^{\eps r/(10D)} \leq e^{-c_1 r^3 \eps /(20D n^2)},\]
    which is upper bounded by $1/(2n^2)$, whenever $r \geq 20 n^{2/3} \cdot \left[D \log n / (c_1\eps) \right]^{1/3}$.

    \medskip

    We now recall assuming that $n \geq c \cdot D \log n /\eps$, using the right constant $c$ (as function of the constant $c_1$). For this case, we also get $r \geq  10 D\log n / \eps$, for which the case of too few noisy rounds is bounded by $e^{-0.4\eps r / D} \leq e^{- 3 log(n)} \leq 1/(2n^2)$.

    By union bound we get the following: using $r = 20 n^{2/3} \cdot \left[D \log n / (c_1\eps) \right]^{1/3}$, for any $D \geq c\cdot \eps n / \log n$, and any realization of $U_r, V_r$, the message has been passed to $V_r$ within the $r$ middle rounds with probability at least $1 - n^{-2}$, and so after $R = 3r$ rounds, $v$ received the message with the same probability. Using union bound over all vertices in the network, we conclude see that $R$ rounds suffice for flooding with probability at least $1-1/n$.
    \hide 
    {
    To conclude, using $r = 20 n^{2/3} D^{1/3}/ (c_1\eps)^{1/3}$: for any realization of $U_{r}, V_{r}$, the message has been passed to $V_r$ during the middle rounds (and therefore to $v$ by the end of all rounds) with probability at least $1/2$.

    By repeating the process for $2\log n$ times, we have failure probability of $n^{-2}$ for each non-source vertex $v$. A  union bound over all such vertices implies that $R\log n$ rounds are enough to flood to the whole network with probability at least $1-1/n$.
    }
\end{proof}

Adjusting the same argument for an \emph{adaptive} network, using only two phases (same as Theorem~\ref{thm: fractional adaptive UB}), would give an upper bound of $O(n \cdot \sqrt{D \log n / \eps})$, which in this case does not improve upon the trivial $O(n)$ promised by connectivity alone.

We continue to show that this is tight, and in the proportional noise model, the changes in the network can be so subtle that they would invoke little noise and the flooding would barely speed up at all, asymptotically. We show the following lower bound, and note that the graph sequence constructed in the proof has constant $D$, where it is hardest to manipulate.

\begin{theorem}
\label{thm: FP - adaptive LB}
    Fix $\epsilon \leq 1/5$.
    There exists an $\eps$-proportionally smoothed dynamic graph with an adaptive adversary, where a flooding process must run for $\Omega(n)$ rounds in order to succeed with probability at least $1/2$.
\end{theorem}

In the current, adaptive model, the spooling graph no longer stands as a slow-flooding graph. Changing the center of the uninformed star takes $\Omega(n)$ edge changes in early rounds, and the proportional noise invoked by these changes would be devastating (e.g., for~$\eps = \Omega(1)$).

Instead, in the following proof, the adversary relies on its adaptivity in order to maintain control over the network. This is done using a sequence of graphs that require only $O(1)$ changes at each round.

\begin{proof}
    We modify the spooling graph in a way that only few edges are changed at each step. Let $G_1,\dots, G_{n-1}$ be the sequence of graphs over $[n]$, each having the edge set
    \[E_i = \set{(1,j)\mid j\leq i} \cup \set{(j,n)\mid j \geq i}.\]
    Note that $G_i$ is connected through the vertex $i$. These graphs are essentially two star graphs, one around $1$ and one around $n$ that are connected through one other vertex (a different one at each time frame). We associate the star centered in $1$ with the set of \emph{informed} vertices and the star centered in $n$ with the set of \emph{uninformed} vertices (as $1$ is the the vertex from which the flooding begins).

    Wishfully, we would like to say that at each time connectivity is preserved via the vertex $i$, but as soon as it receives the message, it is being cut off from the uninformed part of the graph. However, the added noise would impair this argument. As we wish to use a low-maintenance dynamic graph, we cannot afford to replace the center. Therefore, we surgically cut off from the uninformed star all the vertices that learned the message via edges added by noise.
    This enforces us to incorporate into the lower bound the concept of \emph{adaptivity}.
    Let $G_1$ be as before. we define the rest of the graphs in the sequence adaptively. Let $I_i$ be the set of vertices that received the message by the end of round $i$, and let $\bar{I}_i = V / I_i$, the set of uninformed vertices at the same time. We define the connecting vertex of the next graph as $u_{i+1} = argmin(\bar{I}_i)$. The graph $G_{i+1}$ in the sequence is then defined using
    \[E_{i+1} = \set{(1,j)\mid j\in I_i} \cup \set{(j,n)\mid j\in\bar{I}_i} \cup \set{(1,u_i)}.\]

    Consider the first $R = \delta n$ rounds of the process, for small enough $\delta$. We show that if no edge of $I_i \times n$ was added at some round $i+1$, which happens with small probability, we have:
    \begin{enumerate}
        \item The number of required changes at each iteration is at most constant at each round.
        \item At each point in time $\size{I_j} < 2j$.
    \end{enumerate}

    For the first bullet, note that $2 \leq \size{G'_{i-1}\oplus G_{i}} \leq 5$. Indeed, at the first round we only change $2$ edges, replacing $(u_1,n)$ by $(1,u_2)$.
    By induction, for any round $i > 1$, if noise was invoked in the previous round, then since $5 \eps < 1$, at most one extra edge was changed in $G'_{i-1}$. We have 3 options for the toggled edge
    \begin{itemize}
        \item If an edge was added by noise, connecting an informed vertex with some uninformed vertex $w \neq n$, we remove the edge, cut off the edge $(w,n)$ as well, and add $(1,w)$ instead.\footnote{Leaving $w$ connected through the vertex $1$ and not the vertex who sent him the message keeps our low diameter intact. A more freely defined graph with no diameter guarantees could simply remove $(w,n)$ and leave it connected to the informed spool.}
        \item If an edge was added by noise, connecting two uninformed vertices, we cut it off immediately (rather than analysing its affect later, when one of them  becomes informed).
        \item If any other edge was connected (or disconnected) by noise, we simply revert the change, to keep our graph in line with its definition.
    \end{itemize}

    Either way, we need to add two more changes (if not done already above) which consist of removing the edge $(u_{i-1},n)$ and adding the edge  $(1,u_i)$ instead. In total, the number of changes in the next round stays between $2$ and $5$, proving the induction step.

    For the second bullet, note that the amount of noisy edges produced at each round is at most $1$, which means that at most $2$ new vertices learns the message in each round, including the one connecting the graph. This means that $\size{I_i} \leq \size{I_{i-1}} + 2$, or alternatively $\size{I_j} \leq 1 + 2j \leq 3j$.

    As both the above claims hold whenever $n$ has not yet learnt the message, we are left to show that with high probability, within $R = \delta n$ rounds, this is indeed the case. We union bound over all $R$ rounds: as long as it has not yet happened yet, the probability of $n$ getting the message at the next round is that of a noisy edge connecting $I_i$ to the vertex $n$.

    We apply Lemma~\ref{lem: UB adding from a set of edges} for each round $i < R$, using $S_i = I_i \times \set{n}$, and assuming $n$ is not yet informed. Thus, the probability is at most:
    \[c_2 \size{S_i} / n^2 \leq c_2 \cdot 3R / n^2\]
    since $\size{S_i} = \size{I_i} \leq 3i \leq 3R$, and the number of noisy edges is at most $1$. Applying a union bound over all $R$ rounds, the probability of vertex $n$ being prematurely informed is at most
    \[3c_2 \cdot R^2 / n^2 \leq 3c_2 \delta^2.\]
    For small enough $\delta$, this is at most $0.1$, so with probability at least $0.9$ after $R$ rounds we indeed have at most $3R < n$ informed vertices, which means the flooding is not yet over.
\end{proof}

\subsection{Targeted noise}
For targeted noise this new phenomenon is surpassed by the limited ability of the adversary to make changes. We can think of targeted noise as some sort of ``slow down'', as if the network repeatedly tries to modify some of the edges, eventually succeeding, but not before a number of rounds has passed.
In this last model we show just how strong the waiting game can be: for a graph with constant static diameter (which makes the waiting game obsolete), flooding will take $O(\log n)$ rounds, with high probability. However, for larger value of $D$, the same analysis would quickly fail. For $D = \Theta\left(\sqrt{\log n}\right)$ we get the trivial bound of $O(n)$. We finish by showing an explicit construction with static diameter $D = \Theta\left(\log n\right)$ that relies strongly on the waiting game and admits a lower bound of $\Omega(n)$ rounds.

\begin{theorem}
\label{thm: flooding targeted UB}
For any $\epsilon$-targeted smoothed dynamic graph with static diameter $D$, flooding can be done in $O(D \log n /\eps^{D^2})$ rounds, w.h.p.
\end{theorem}

Note that for $D = o(\sqrt{\log_{1/\eps} n})$, this improves upon the trivial fact that $n$ rounds are enough (as at each round, one new vertex must be informed, since the graph is connected). Specifically, for a constant static diameter, we get $O(\log n)$ rounds.

\begin{proof}
Fix the starting vertex $v_0$ which is informed.
First, we show that the probability of a vertex $v\neq v_0$ staying uninformed after $D$ rounds is at most $1-\eps^{D^2}$.

Consider $D$ consecutive graphs $G_1,\ldots, G_D$ in the smoothed graph. As the static diameter is $D$ in every round, there exists a path $P_v$ from $v_0$ to $v$ in $G_1$, whose length is at most $D$. Since we deal with targeted smoothing, each edge of $P_v$ that exists in a graph $G_i$ for some $1\leq i<D$ exists in $G_{i+1}$ probability either $1$ or $\eps$. So, for each such $i$, if all the edges of $P_v$ exist in $G_i$ then they all exist in $G_{i+1}$ with probability at least $\eps^D$.
Hence, the probability of the path $P_v$ existing in all $D$ graphs is at least $(\eps^D)^D = \eps^{D^2}$.

Fix a positive integer $t$, and consider $t$ consecutive sequences of $D$ graphs each.
On each sequence, we apply the above claim, and conclude that the probability of $v$ staying uninformed after theses $tD$ rounds is at most
\[(1-\eps^{D^2})^t \leq e^{-t\cdot \eps^{D^2}}.\]

For $t = (c+1)\log n \cdot (1/\eps)^{D^2}$ sequences, the probability of $v$ not being informed after $tD$ rounds is at most $n^{-(c+1)}$. A union bound over all $n-1$ vertices that need to be informed implies that $tD$ rounds suffice to fully flood the network with probability at least $1-n^{-c}$.
\end{proof}

The above theorem implies that if the static diameter of all graphs in the sequence is small, roughly $O(\sqrt{\log n})$, flooding is fast.
Next, we show that this is almost tight: if the diameter is in $\Omega(\log n)$, flooding cannot be done faster than in non-smoothed graphs.

\begin{theorem}
\label{thm: flooding targeted LB}
For every constant $0<\epsilon<1$,
there is a value $D\in \Theta(\log n)$
and an $\epsilon$-targeted smoothed dynamic graph
such that with high probability,
the diameter of the graph is $D$ and flooding on it takes $n-1$ rounds.
\end{theorem}

In order to prove this theorem, we present the \emph{dynamic cassette graph} (see  Figure~\ref{fig: cassette}).
Fix $\epsilon$ and $n$ as in the theorem statement, and let $t= \floor{c\log_{1/\eps} n}$ for a constant $c$ of choice.
The dynamic cassette graph on vertices $V = \set{0,\dots, n-1}$ is the dynamic graph $\Hcal= \set{G_1, \dots, G_n}$, where $G_i=(V,E_i)$ is defined by
\begin{align*}
E_i =   & \set{(j,j+1)\mid 0\leq j< n-1} \cup \\
        & \set{(0,jt)\mid 1\leq j\leq\floor{(i-1)/t}} \cup
          \set{(jt,n-1)\mid \floor{(i-1)/t}+2\leq j\leq \floor{(n-2)/t}}
.
\end{align*}

\begin{figure}
\centering
  \includegraphics[scale=0.3]{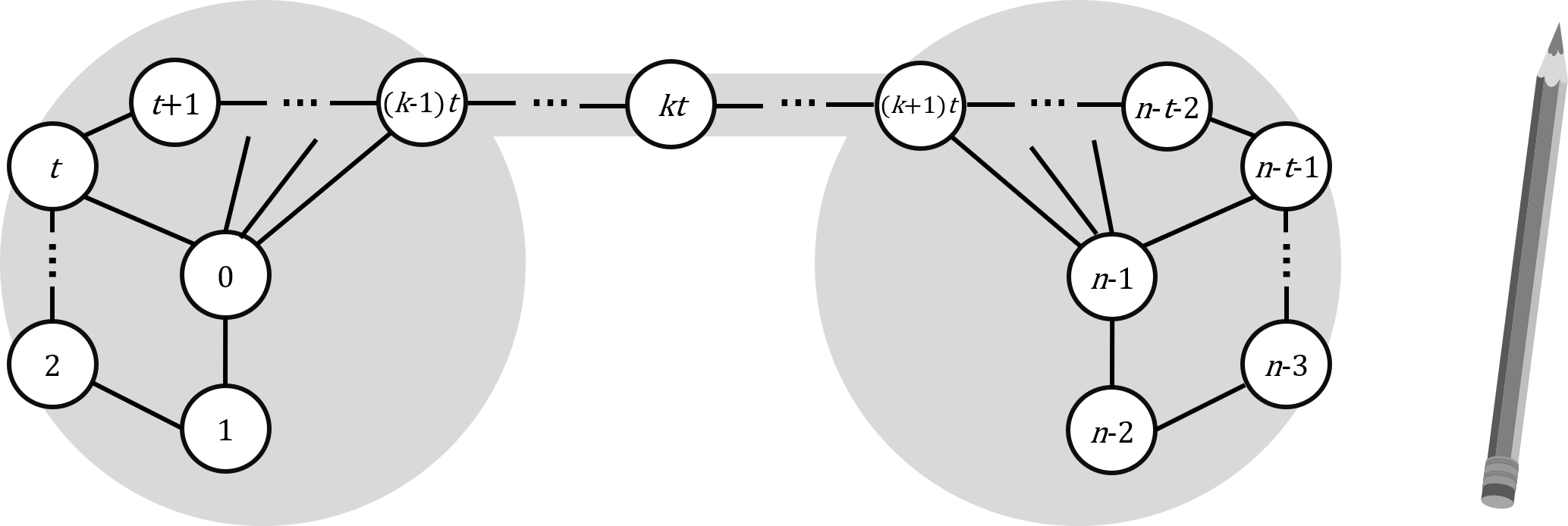}
  \caption{The cassette graph, $G^t_j$, where $j=kt$ and $n$ is some multiple of $t$.
  }
  \label{fig: cassette}
\end{figure}

This graph is the path on $n$ vertices, with some additional edges connecting the first and last vertices to vertices in the set $\set{jt \mid 1\leq j\leq (n-2)/t}$; these will be referred to as \emph{shortcut vertices}, and the additional edges to them, \emph{shortcut edges}.
At the first graph, $G_1$, all shortcut vertices but the first are connected to the last vertex, $n-1$.
Then, one by one, the shortcut vertices disconnect from $n-1$, and soon after --- connect to $0$.
At each time interval $[(j-1)t+1,jt]$, all the shortcut vertices with index strictly smaller than $jt$ are connected to the vertex $0$, and all those with index strictly higher than $jt$ are connected to~$n-1$.

Consider the \emph{smoothed cassette graph} $\Hcal'$, i.e., the $\epsilon$-targeted smoothed dynamic graph derived from $\Hcal$.
The dynamic graph $\Hcal'$ can be interpreted as undergoing the following process:
during each time interval $[(j-1)t+1,jt]$, the adversary repeatedly tries to add a new edge $(0,(j-1)t)$, and remove the edge $((j+1)t, n-1)$.
The targeted noise creates a slowdown that might prevent this from happening right away,
yet for the right value of $t$, both changes indeed happen by the end of the time interval w.h.p. We state the following claim and direct the reader to Appendix~\ref{Sec: appendix thm proof} for a full proof.



\begin{claim}
\label{claim: cassette specific shortcut edges whp}
For each $2\leq j\leq (n-2)/t$, the smoothed graph $G'_{jt}$
does not contain the edge $(0,(j-1)t)$ with probability at most $n^{-c}$,
and contains the edge $((j+1)t, n-1)$ with the same probability.
\end{claim}

If the edge $(0,(j-1)t)$ exists in the smoothed graph $G'_{jt}$
then it also exists in all later graphs, $G_{j'}$ with $j'>jt$.
Similarly, if the edge $((j+1)t, n-1)$ does not exist in this graph, it also does not appear in later graphs.
A union bound thus extends the last claim as follows.

\begin{claim}
\label{claim: cassette all shortcut edges whp}
For each $2\leq j\leq (n-2)/t$, the smoothed graph $G'_{jt}$ and all later graphs
contain the edge $(0,(j-1)t)$ with probability at least $1-n^{-c+1}$,
and all these graphs do not contain the edge $((j+1)t, n-1)$ with the same probability.
\end{claim}

Using this claim, Theorem~\ref{thm: flooding targeted LB} can easily be proven.
\begin{proof}[Proof of Theorem~\ref{thm: flooding targeted LB}]
Consider the smoothed dynamic cassette graph $\Hcal'$. We start by analyzing its diameter.

Let $G'_{j'}$ be a graph in $\Hcal'$, and pick a $j$ such that $jt\leq j'< (j+1)t$.
By Claim~\ref{claim: cassette all shortcut edges whp}, the graph $G'_{jt}$ and all later graphs contain the edge $(0,(j-1)t)$ with probability at least $1-n^{-c+2}$.
In addition, all the graphs $G_1,\ldots,G_{(j+1)t-1}$ contain the edge $((j+2)t, n-1)$ (and note that this is not a probabilistic claim).

The distance between every two shortcut vertices is $t$, so the distance from every vertex in the graph to the closest shortcut vertex is at most $t/2$.
Each shortcut vertex is directly connected to either $0$ or $n-1$ w.h.p., except for $jt$, who is connected to both by a $t+1$ path.
Finally, between the vertices there is a path of length $2t+2$, through $(j-1)t$ and $(j+1)t$.
Let us bound the length of a path between two vertices $i,i'$ (with upper bounds on each part): from $i$ to its closest shortcut vertex ($t/2$ hops), to $0$ or $n-1$ ($t+1$ hop), maybe to the other vertex between $0$ and $n-1$ ($2t+2$ hops), to the shortcut vertex closest to $i'$ ($t+1$ hop), and to $i'$ ($t/2$ hops).
This sums to a path of length at most $5t+4=\Theta(\log n)$.
A more detailed analysis can reduce this to roughly $3t$ hops.

For the flooding time, we use Claim~\ref{claim: cassette all shortcut edges whp} again, but now for the edges $((j+1)t, n-1)$ that do not appear $G'_{jt}$ and all later graphs w.h.p.
A simple induction on $j'=0,\ldots,n-1$ shows that after $j'$ rounds, i.e. in graph $G_{j'}$, only vertices $0,\ldots,j'$ are informed.
The base case is trivial. For the step, the only edges connecting informed vertices to uninformed vertices are $(j',j'+1)$, and edges from $0$ to shortcut vertices, which have the form $(0,tj)$ with $tj\leq j'$ --- this yields from the construction and always holds.
The only other type of possible edges connecting informed and uninformed vertices are of the form $(jt,n-1)$, with $jt\leq j'$. However, the claim implies that w.h.p., by round $j'$ none of these edges exist.
Hence, before round $n-1$ not all vertices are informed, and flooding takes $n-1$ rounds w.h.p.
\end{proof}


\hide{end of real content}

\bibliography{refs}


\bigskip
\centerline{\Large\bf Appendix}

\appendix
\section{Omitted Proofs}
\label{Sec: appendix thm proof}
\begin{theorem-repeat}{thm: fractional LB 1}
   Fix $1 \leq k \leq n/16$ (not necessarily an integer). For any $k$-smoothed dynamic graph, for the flooding process to succeed with probability at least $1/2$, it must run for $\Omega\left(\min\set{n/k, n^{2/3}/k^{1/3}}\right)$ rounds.
\end{theorem-repeat}

\begin{proof}
    We start by re-defining the spooling graph from \cite{DFGN18}:
    at time $i$, the edges are $\set{(j,i)}_{j=1}^{i-1}$, a star around the vertex $i$ (this will be the informed spool)
    and edges $\set{(i+1,j)}_{j=i+2}^{n}$, a star around the vertex $i+1$, which will be the right spool. The spools are than connected using $(i, i+1)$

    Next, we consider the first $R$ rounds of flooding and show that for small enough $R$, with at least constant probability, the message has not reached all the players.
    We observe the vertices $1,\dots,R$ that are going to be the heads of the uninformed spool.

    Note that the right spool is not called ``the uninformed spool'' for a reason: some vertices in it might get the message through the process. But every time they send the message to the head of the spool, it is being ``snatched'' to the informed spool before it gets to spread the message to the whole right spool. (we stress this observation, as it is going to serve us later in the text as well).
    We must deal with two bad cases:
    \begin{enumerate}
        \item Some relevant head $j < R$ gets the message in round $l < j-1$, meaning that he spreads the message to everyone on round $j-1$ when he becomes the head of the right spool.
        \item At any point in the process, more than $O(R)$ vertices in the set of ``right spool player'' ($\set{R+1,\dots,n}$) learned the message via a noisy edge that was added to the graph.
    \end{enumerate}

    Note that the second bad case is not about the end of the process (otherwise, it would suffice to ask that not \emph{all} vertices learnt the message). Rather, we ask for less than $O(R)$ vertices, as we need this to guarantee the first bad case only happens in rare occasions.

    To that end, we recall the players are numbered $1,\dots, n$. We denote by $I_r$ the set of informed vertices after round $r$ of the flooding.

    We carry on using a careful induction on the two claims combined. The induction is probabilistic, stating that each step has a small probability to fail us, and if it does not - it leaves us in a rather good state for the next step. We finish up by union bounding over all $R$ rounds.

    Let us denote by $A_i$ the good event of no future head vertex getting the rumor at round $i$ (and $A_{\leq i}$ for $A_1 \bigcap \dots \bigcup A_i$)
    .
    Formally we prove by induction that w.h.p: (1) $A_i$ occur at round $i$; (2) The expected size of $I_i$, conditioned on $A_{\leq i}$ relates to the expected value of $I_{i-1}$.

    For the first round $i=1$, only the starting vertex $1$ knows the message and the claims hold trivially.
    Let us assume that both claims hold for round $i$, we show that other than with small probability, both also hold for round $i+1$:

    We assume that $A_{\leq i}$ occurs (meaning in particular that $I_i \bigcap [R] = [i]$). Therefore, we can use Lemma~\ref{lem: UB adding from a set of edges} with $S = I_i \times ([R] \setminus [i+1])$. We then know that $A_{i+1}$ (and therefore $A_{\leq i+1}$) occurs with probability at least
    \[1 - c_2 k \size{S} / n^2 \leq c_2 k \size{I_i} R / n^2.\]

    we now focus on a vertex that is not a future head $v \in I_i \setminus [R]$, and note that within $k$ noisy edges in round $i+1$, similarly to Lemma~\ref{lem: UB adding from a set of edges}, the probability of adding an edge from the set $I_i \times \set{v}$, \emph{given the event $A_i$}, is at most
    \[c_2 k \size{I_i} / (n^2 - \size{I_i}\cdot R) \leq 2c_2 k \size{I_i} / n^2\]
    Where the statement is evident from the proof of Lemma~\ref{lem: UB adding from a set of edges} (in~\cite{DFGN18}), and the inequality is due to $I_i \cdot R \leq n^2/2$ for any realization of $I_i$.

    Therefore, if we denote by $\ind{v,i}$ the probability of $v$ being added at round $i$, we have
    \[\Expc{\ind{v,i}  | G_{\leq i}} \leq 2c_2 k \size{I_i} / n^2.\]
    Next, we sum up over all possible $v$'s (at most $n$ of them), to get that for any specific realization of $I_i$:
    \[\Expc{I_{i+1} | A_{\leq i+1}} \leq I_i + n\cdot \frac{2c_2 k \size{I_i}}{n^2} + 1\]
    where the expectation is taken over round $i+1$ alone, the $1$ comes from the guaranteed new informed vertex (the new head), and we use the guarantee of $A_i$ that the previous head was not informed.

    Taking an expectation over the noise in previous $i$ rounds as well, we get:
    \[\Expc{I_{i+1} | A_{\leq i+1}} \leq \Expc{I_i | A_{\leq i}}\cdot (1 + \frac{2c_2 k}{n}) + 1\]
    Now both expectations are over all noise (note that future noise cannot affect the value of~$I_i$).

    We denote $L_i := \Expc{I_i | A_{\leq i}}$ to analyse this multiplicative-additive progression.
    Up until round $R = n/(4c_2 k)$, this progression is at most $2$ at each step, giving us $L_{R} \leq 2R + 1 \leq O(n/k)$. We do not deal with rounds of dominant multiplicative factors, for the sake of our union bound. We get that after $R$ rounds, the expected number of players that know that rumor is far smaller than $n$.

    We now enter the probabilities to have:
    \begin{enumerate}
        \item Each inductive step had failure probability of at most $c_2 k\size{I_i}R / n^2$, where $I_i$ is smaller than $O(i) \leq O(R)$. Union bounding over all $R$ rounds, our failure probability is at most $10 c_2 k R^3 /n^2$.
        \item Using Markov inequality, since the expected number of informed vertices at round $R$ is $O(n/k)$, and so the probability at round $R$ of having more than $n/2$ informed vertices is at most $(n/k) / (n/2) \leq O(1/k)$.
    \end{enumerate}
    We note that the second failure probability is negligible, and the first (union bound over all steps) is at most $1/10$ whenever we have $R \leq n^{2/3} / (100 c_2 k)^{1/3}$.
    Overall, we need both the progression to be additive, and the union bound to work, which means the flooding fails with high probability whenever
    \[R = O\left(\min\set{\frac{n}{4c_2 k}, \frac{n^{2/3}}{(100 c_2 k)^{1/3}}}\right),\]
    completing the proof.
\end{proof}

\begin{theorem-repeat}{thm: fractional LB 2}
    Fix $0 < k < 1$. For any $k$-smoothed dynamic graph, for the flooding process to succeed with probability at least $1/2$, it must run for $\Omega\left(\min\set{n, n^{2/3}/k^{1/3}}\right)$ rounds.
\end{theorem-repeat}

\begin{proof}
    The proof follows an easier argument than the previous one, as no exponential growth in the number of informed players can ever happen. This means that our only concern is that of a future head prematurely learning the message.

    Let us focus on the first $R$ rounds. Using a standard Chernoff bound, we know that with high probability at most $10kR$ rounds contain any noise at all.
    All other rounds add exactly one new informed player per round, and the noisy rounds can add at most $1$ more player using a noisy edge.
    Hence, as long as $A_i$ keeps occurring (the events of no future head receiving the message prematurely), the inequality $I_i \leq 3i$ holds.

    However, at each non-noisy iteration, the probability of $G_i$ failing is exactly $0$, and at each iteration \emph{with noise}, the probability of $A_i$ failing is at most
    \[c_2 \size{I_i} R / n^2 \leq 3 c_2 i R / n^2 \leq 3 c_2 R^2 / n^2\]
    Applying a union bound over all $10 k R$ noisy rounds, we get the event $A_{\leq R}$ fails with probability at most
    \[30c_2 k R^3 / n^2,\]
    which is smaller than $1/10$ whenever $R = n^{2/3} / (300c_2 k)^{1/3}$.

    But this would mean that with high probability indeed $I_{i+1} \leq I_i + 2$ for any $i < R$, and thus no more than $3R$ players are informed after $R$ rounds.
    For any $k \geq 1/(10,000 c_2 n)$, we have $R < n/3$, or $3R < n$, concluding the protocol failed.

    On the other hand, for $k \leq  1/(10,000 c_2 n)$, considering $\delta n$ rounds for small enough $\delta$, we get that with constant probability no round invoked any noise, which means the flooding will take $\Omega(n)$, as is the case with no noise.
\end{proof}

    \begin{claim-repeat}{claim: expected informed vertices}
        When taking expectation over the noise we add, we have
        $$\Expc{\size{I_{i+1}}} \leq (1+c_2 k/n)\cdot \Expc{\size{I_{i}}} + 1 .$$
    \end{claim-repeat}
    Note that $I_i$ does not depend on the noise of rounds $i+1,\dots,R$, and so the left side takes expectation over $i+1$ rounds of noise, and the right side over $i$ rounds.

    \begin{proof}
        Indeed, just before the (i+1)$^{th}$ iteration, there are exactly $\size{I_i}$ informed vertices. Therefore, each uninformed vertex $v\in \bar{I}_i \setminus \set{u_i}$ has probability of at most $c_2 k \size{I_i} / n^2$ of obtaining a straight edge to an informed vertex (Using $S = I_i \times \set{v}$). As all this vertices were disconnected from $I_i$ before the noise, this is their only way to receive the message at this round.
        For $u_i$ we use the trivial upper bound of $1$ (in fact, it is even smaller: it has a small probability to get disconnected).
        For each $v\in \bar{I}_i$ we denote by $\ind{v}$ the event that $v$ learnt the message at round $i+1$.
        For each realization of $I_i$ after $i$ rounds, we have
        \begin{align*}
            \Expc{\size{I_{i+1}}}
            & = \size{I_i} + \sum_{v \in \bar{I}_i} \ind{v}\\
            & \leq \size{I_i} + 1 + (\size{\bar{I}_i}-1) c_2 k \size{I_i} / n^2\\
            & \leq (1 + c_2 k / n) \size{I_i} + 1
        \end{align*}
        where the expectation is over the noise in round $i+1$, and the second inequality uses the fact $\size{\bar{I}_i} \leq n$.

        Using linearity of this term, and taking expectation over the noise of the first $i$ rounds in both sides, the claim is proven.
    \end{proof}

\begin{claim-repeat}{claim: cassette specific shortcut edges whp}
For each $2\leq j\leq (n-2)/t$, the smoothed graph $G'_{jt}$
does not contain the edge $(0,(j-1)t)$ with probability at most $n^{-c}$,
and contains the edge $((j+1)t, n-1)$ with the same probability.
\end{claim-repeat}

\begin{proof}
We prove the first part, and the second one is dual using a similar argument. Fix $j$, $2\leq j\leq (n-2)/t$. Consider the edge $e = (0,(j-1)t)$, and denote by $A_i$ the event of $G'_{i}$ not containing the edge $e$. Under these notations, our goal is to upper bound $\Pr[A_{jt}]$.

First, note that the edge $e$ does not appear in any graph (before and after the smoothing) in the first $(j-1)t$ rounds, and so for $i \in \set{1, \dots, (j-1)t}$ we have:
\[\Pr[A_i] = 1.\]
Recall that each iteration imposes smoothing with respect to the previous (already smoothed) graph. Starting with $G_{(j-1)t+1}$, the non-smoothed graphs include $e$, and the targeted smoothing might disrupt this addition. However, once the addition happens, the smoothing process does not target it any more. So, for $i \in \set{(j-1)t+1, \dots, jt}$:
\[
\Pr[{A_i} | A_{i-1}] = \eps \text{, and }  \Pr[{A_i} | \bar{A}_{i-1}] = 0,
\]
which implies
\[\Pr[A_i] = \eps \cdot \Pr[A_{i=1}].\]
Starting with $\Pr[A_{(j-1)t}] = 1$, we finally conclude $\Pr[A_{jt}] = \eps^t$.

Fix $j$, $2\leq j\leq (n-2)/t$.
For $j=2$, note that the edge $(0,t)$ is not present in any graph so the claim regarding it trivially holds.

Consider the $t$ non-smoothed graphs $G_{(j-1)t+1},\ldots,G_{jt}$,
and note that all of them contain the edge $(0,(j-1)t)$.
In each of the corresponding smoothed graphs, $G'_{(j-1)t+1},\ldots,G'_{jt}$,
the edge $(0,(j-1)t)$ does not appear only if it did not appear in the previous graph,
in which case it does not appear in the current graph with probability $\epsilon$.
Hence, the probability this edge does not appear in the last graph $G'_{jt}$ is at most $\eps^t=n^{-c}$.
Similarly, the edge $((j+1)t, n-1)$ does not appear in all the graphs $G_{(j-1)t+1},\ldots,G_{jt}$,
so it appear in the last graph $G'_{jt}$ with probability at most~$n^{-c}$.
\end{proof}

\end{document}